\documentclass[journal,11pt,a4paper,onecolumn]{IEEEtran}
\usepackage{graphicx}
\usepackage{amssymb,amsbsy}
\usepackage{epstopdf}
\usepackage{amsmath,amsthm}
\usepackage{enumerate}
\usepackage{eufrak}
\usepackage{cite}
\usepackage{mathcomp}
\usepackage{supertabular}
\usepackage{longtable}
\usepackage{stmaryrd}
\usepackage{url}
\usepackage{color}
\usepackage{rotating}
\usepackage{float}

\usepackage{blkarray, multirow}
\def\Big#1{\makebox(0,-10){\huge{#1}}}

\usepackage{array}
\newcolumntype{C}[1]{>{\centering\arraybackslash\hspace{0pt}}p{#1}}

\usepackage{array}
\usepackage{mathtools}
\usepackage{multicol}
\usepackage{amsmath}
\usepackage{algorithm}
\usepackage{algorithmic}
\usepackage{amssymb}
\usepackage{xcolor}

\usepackage[all]{xy}
\entrymodifiers={++[o][F-]}

\usepackage{setspace}

\theoremstyle{definition}

\newtheorem{prop}{Proposition}[section]
\newtheorem{thm}[prop]{Theorem}
\newtheorem{cor}[prop]{Corollary}
\newtheorem{lem}[prop]{Lemma}

\newtheorem{exa}{Example}[section]

\newtheorem{rem}{Remark}

\begin{document}

\newcommand{\vA}{{\bf A}}
\newcommand{\vAtilde}{\widetilde{\bf A}}

\newcommand{\vB}{{\bf B}}
\newcommand{\vBtilde}{\widetilde{\bf B}}

\newcommand{\vC}{{\bf C}}
\newcommand{\vD}{{\bf D}}
\newcommand{\vG}{{\bf G}}
\newcommand{\vH}{{\bf H}}
\newcommand{\vI}{{\bf I}}

\newcommand{\vY}{{\bf Y}}
\newcommand{\vZ}{{\bf Z}}

\newcommand{\vJ}{{\bf J}}

\newcommand{\vM}{{\bf M}}
\newcommand{\vN}{{\bf N}}
\newcommand{\vU}{{\bf U}}
\newcommand{\vV}{{\bf V}}
\newcommand{\vT}{{\bf T}}
\newcommand{\vR}{{\bf R}}
\newcommand{\vS}{{\bf S}}

\newcommand{\va}{{\bf a}}
\newcommand{\vb}{{\bf b}}
\newcommand{\vc}{{\bf c}}

\newcommand{\ve}{{\bf e}}
\newcommand{\vh}{{\bf h}}
\newcommand{\vg}{{\bf g}}
\newcommand{\vp}{{\bf p}}

\newcommand{\vu}{{\bf u}}
\newcommand{\vv}{{\bf v}}
\newcommand{\vw}{{\bf w}}
\newcommand{\vx}{{\bf x}}
\newcommand{\vhx}{{\widehat{\bf x}}}
\newcommand{\vtx}{{\widetilde{\bf x}}}
\newcommand{\vy}{{\bf y}}
\newcommand{\vz}{{\bf z}}

\newcommand{\vj}{{\bf j}}
\newcommand{\vzero}{{\bf 0}}
\newcommand{\vone}{{\bf 1}}
\newcommand{\vbeta}{{\boldsymbol \beta}}
\newcommand{\vchi}{{\boldsymbol \chi}}

\newcommand{\tA}{\textrm A}
\newcommand{\tB}{\textrm B}
\newcommand{\A}{\mathcal A}
\newcommand{\B}{\mathcal B}
\newcommand{\C}{\mathcal C}
\newcommand{\D}{\mathcal D}
\newcommand{\E}{\mathcal E}
\newcommand{\F}{\mathcal F}
\newcommand{\G}{\mathcal G}
\newcommand{\M}{\mathcal M}
\newcommand{\HH}{\mathcal H}
\newcommand{\PP}{\mathcal P}

\newcommand{\Q}{\mathcal Q}
\newcommand{\Qb}{\bar{\mathcal Q}}
\newcommand{\Db}{{\bar{\Delta}}}

\newcommand{\pQ}{{\bf p}\mathcal Q}
\newcommand{\pQb}{{\bf p}\bar{\mathcal Q}}

\newcommand{\R}{\mathcal R}
\newcommand{\SSS}{\mathcal S}
\newcommand{\U}{\mathcal U}
\newcommand{\V}{\mathcal V}
\newcommand{\Y}{\mathcal Y}
\newcommand{\Z}{\mathcal Z}

\newcommand{\Pg}{{{\mathcal P}_{\rm gram}}}
\newcommand{\Pgint}{{{\mathcal P}^\circ_{\rm gram}}}
\newcommand{\Pgrc}{{{\mathcal P}_{\rm GRC}}}
\newcommand{\Pgrcint}{{{\mathcal P}^\circ_{\rm GRC}}}
\newcommand{\Pint}{{{\mathcal P}^\circ}}
\newcommand{\Ag}{{\bf A}_{\rm gram}}

\newcommand{\CC}{\mathbb C} 
\newcommand{\RR}{\mathbb R}
\newcommand{\ZZ}{\mathbb Z}
\newcommand{\FF}{\mathbb F}
\newcommand{\KK}{\mathbb K}

\newcommand{\Fnd}{\FF_q^{n^{\otimes d}}}
\newcommand{\Knd}{\KK^{n^{\otimes d}}}

\newcommand{\ceiling}[1]{\left\lceil{#1}\right\rceil}
\newcommand{\floor}[1]{\left\lfloor{#1}\right\rfloor}
\newcommand{\bbracket}[1]{\left\llbracket{#1}\right\rrbracket}

\newcommand{\inprod}[1]{\left\langle{#1}\right \rangle}


\newcommand{\beas}{\begin{eqnarray*}} 
\newcommand{\eeas}{\end{eqnarray*}} 

\newcommand{\bm}[1]{{\mbox{\boldmath $#1$}}} 

\newcommand{\wt}{{\rm wt}} 
\newcommand{\supp}{{\rm supp}} 
\newcommand{\dg}{d_{\rm gram}} 
\newcommand{\da}{d_{\rm asym}} 
\newcommand{\dist}{{\rm dist}} 
\newcommand{\ssyn}{s_{\rm syn}}
\newcommand{\sseq}{s_{\rm seq}}
\newcommand{\nullplus}{{\rm Null}_{>\vzero}}

\newcommand{\tworow}[2]{\genfrac{}{}{0pt}{}{#1}{#2}}
\newcommand{\qbinom}[2]{\left[ {#1}\atop{#2}\right]_q}

\newcommand{\Lovasz}{Lov\'{a}sz }
\newcommand{\etal}{\emph{et al.}}

\newcommand{\citereq}{{\color{blue} [citation required]}}
\newcommand{\todo}[1]{{\color{red} (TODO) #1}}

\title{Local Codes with Addition Based Repair 
\thanks{Part of this work was completed when Han Mao Kiah visited Singapore University of Technology and Design.}}
 \author{
   \IEEEauthorblockN{
     Han Mao Kiah\IEEEauthorrefmark{2},
     Son Hoang Dau\IEEEauthorrefmark{1},
     Wentu Song\IEEEauthorrefmark{3}, 
     Chau Yuen\IEEEauthorrefmark{4}
		} \\
   \IEEEauthorblockA{
	\IEEEauthorrefmark{2}School of Physical and Mathematical Sciences, Nanyang Technological University, Singapore \\   	
		\IEEEauthorrefmark{1}\IEEEauthorrefmark{3}\IEEEauthorrefmark{4}Singapore University of Technology and Design, Singapore\\ 
		Emails: 
		{\it\IEEEauthorrefmark{2}hmkiah}@ntu.edu.sg,
		$\{${\it\IEEEauthorrefmark{1}sonhoang\_dau, 
		\IEEEauthorrefmark{3}wentu\_song,
		\IEEEauthorrefmark{4}yuenchau}$\}$@sutd.edu.sg			
		}
 }
\IEEEoverridecommandlockouts

\maketitle

\begin{abstract}
We consider the complexities of repair algorithms for locally repairable codes
and propose a class of codes that repair single node failures using addition operations only,
or {\em codes with addition based repair}.
We construct two families of codes with addition based repair. 
The first family attains distance one less than the Singleton-like upper bound,
while the second family attains the Singleton-like upper bound.
\end{abstract}


\section{Introduction}

Motivated by practical implementations of erasure coding in large-scale distributed storage systems, 
{\em locality} was identified as an important metric of study independently 
by Gopalan \etal{} \cite{Gopalan.etal:2012}, Oggier and Datta \cite{Oggier.Datta:2011}, and Papailiopoulos and Dimakis \cite{Papailiopoulos.Dimakis:2012}. 
Generally, the locality of a node $i$ refers to the number of other nodes that needs to be accessed in order to recover node $i$.
We refer to the latter set of nodes as the {\em repair set} for $i$ and
observe that having small values of locality allows for the fast repair of single node failures.

This observation has impelled the study and construction of erasure codes with low locality and high minimum distance.
Extending the argument from the classical Singleton bound, Gopalan \etal{} \cite{Gopalan.etal:2012} established a relationship between the parameters: 
code length $n$, code dimension $k$, locality $r$ and minimum distance $d$ 
(see Theorem \ref{thm:singleton}).
Since then, families of optimal codes attaining this bound were constructed (see \cite{Silberstein.etal:2013,Tamo.etal:2013,Song.etal:2014,Ernvall.etal:2014,Tamo.Barg:2014}).

However, the complexity of the repair algorithms for these code families is often not explicitly discussed.
Nevertheless, we observe that the repair algorithms usually involve multiplication and division in some finite field -- 
operations that are costly to implement in software and hardware.
This cost is exacerbated when the underlying finite field is of a large order.
In the literature of coding for disk arrays, this computational complexity is reduced by replacing finite field arithmetic by simple bit-wise operations. 
For example, Blaum and Roth \cite{Blaum.Roth:1993} proposed a construction of array codes based on the ring of polynomials with binary coefficients. 
More recently, Hou \etal{} \cite{Hou.etal:2013} and Shum \etal{} \cite{Shum.etal:2014} 
proposed classes of regenerating codes that utilize only XOR operations and bit-wise shifts.

Following this line of study, we propose a simple and efficient method of repair: 
when a node fails, we compute the sum of the nodes in its repair set to obtain $S$, 
and replace the failed node by $-S$.
Hence, our repair algorithm uses {\em only addition} operations in a finite field and 
when the field is a binary extension field (as with most storage applications), it utilizes {\em only XOR} operations.
In this paper, we construct two families of linear codes with good locality properties 
that utilizes this algorithm to repair {\em all} its nodes.
The first family achieves {distance one less} than the value in \eqref{eq:singleton},
while the second family achieves exactly the distance, albeit under certain conditions.
These results hence indicate that the extra requirement of addition based repair may be achieved 
{\em without significant loss in erasure correcting capability}.

\begin{exa}\label{exa:motivating}%
Consider a linear code $\C$ of length twelve and dimension six defined over $\FF_{13}$.
Let the information symbols be  $x_1,x_2,\ldots,x_6$ and the generator matrix of $\C$ be
\[\vG=
\left(\begin{array}{cccc|cccc|cccc}
1 & 0 & 0 & 12 & 0 & 0 & 0 & 0 & 7 & 8 & 10 & 1 \\
0 & 1 & 0 & 12 & 0 & 0 & 0 & 0 & 8 & 2 & 5 & 11 \\
0 & 0 & 1 & 12 & 0 & 0 & 0 & 0 & 5 & 3 & 12 & 6 \\
0 & 0 & 0 & 0 & 1 & 0 & 0 & 12 & 1 & 6 & 2 & 4 \\
0 & 0 & 0 & 0 & 0 & 1 & 0 & 12 & 5 & 7 & 8 & 6 \\
0 & 0 & 0 & 0 & 0 & 0 & 1 & 12 & 7 & 11 & 9 & 12
\end{array}\right).
\]

We label the nodes with $1,2,\ldots, 12$ and
 consider the information stored in the last three nodes 9, 10, 11, 12:
\[\begin{array}{rr rrr rrrrrrrrrrrr}
c_{9}\, = &  7 x_{1} &+\, 8 x_{2} &+\, 5 x_{3} &+\, 1 x_{4} &+\, 5 x_{5} &+\, 7 x_{6} ,\\
c_{10}\, = &  8 x_{1} &+\, 2 x_{2} &+\, 3 x_{3} &+\, 6 x_{4} &+\, 7 x_{5} &+\, 11 x_{6} ,\\
c_{11}\, = &  10 x_{1} &+\, 5 x_{2} &+\, 12 x_{3} &+\, 2 x_{4} &+\, 8 x_{5} &+\, 9 x_{6} ,\\
c_{12}\, = &  1 x_{1} &+\, 11 x_{2} &+\, 6 x_{3} &+\, 4 x_{4} &+\, 6 x_{5} &+\, 12 x_{6} .
\end{array}\]

Observe that $c_9=-(c_{10}+c_{11}+c_{12})$.
Hence, when node $9$ fails, 
we repair the node by contacting nodes $10$, $11$ and $12$,
and compute $c_9$ via addition operations over $\FF_{13}$. 
We check that this code has distance six,
thus attaining the upper bound \eqref{eq:singleton}.

%
\end{exa}

The paper is organized as follows. We first describe our main results in Section \ref{sec:main} and 
provide the explicit construction of these code families in Section \ref{sec:proof}. 
In Section \ref{sec:comparison}, we compare our codes with known families of codes with good locality properties., 
namely, Pyramid codes \cite{Huang.etal:2012} and Tamo-Barg codes \cite{Tamo.Barg:2014}.
In particular, we provide an analysis on the repair complexities for these codes.

\section{Main Results}\label{sec:main}

Let $[n]\triangleq \{1,2,\ldots,n\}$.
Consider a linear $[n,k,d]$ code $\C$ over $\FF_q$ of length $n$ with $k$ {\em information} nodes, $n-k$ {\em parity check} nodes and distance $d$.
A node $i\in[n]$ has {\em locality} $r_i$ if there exists a {\em repair set} $J_i=\{j_1,j_2,\ldots,j_{r_i}\}\subseteq [n]\setminus\{i\}$ of size $r_i$ 
and a {\em repair function} $\phi_i:\FF_q^{r_i}\to\FF_q$ such that
$\phi_i(c_{j_1},c_{j_2},\ldots,c_{j_{r_i}})=c_i \mbox{ for all }\vc\in \C$.
The code $\C$ has {\em information locality} $r$ if all information nodes have locality at most $r$,
and has {\em all-symbol locality} $r$ if all nodes have locality at most $r$.


We are interested in repair functions that utilizes only addition operations. 
Formally, a code $\C$ has {\em addition based repair} if for each node $i\in[n]$, 
the repair function defined on the repair set $J_i=\{j_1,j_2,\ldots,j_{r_i}\}$ is given by 
\[\phi_i(c_{j_1},c_{j_2},\ldots,c_{j_{r_i}})=-(c_{j_1}+c_{j_2}+\cdots+c_{j_{r_i}}) .\]


Hence, our objective is to construct a code with addition based repair and good locality properties, 
while maximizing its minimum distance. 
To benchmark the performance of our codes, 
we use the following Singleton-like upper bound on the distance of codes with information locality $r$.


\begin{thm}[{\cite[Th. 5]{Gopalan.etal:2012}}]\label{thm:singleton}
Consider an $[n,k,d]$ linear code with information locality $r$. Then
\begin{equation}\label{eq:singleton}
d\le n-\ceiling{\frac kr}-k+2.
\end{equation} 
\end{thm}

Observe that the above bound applies to linear codes with all-symbol locality $r$ and 
in the same paper, Gopalan \etal{} provided a structure theorem for these codes \cite{Gopalan.etal:2012}.
In particular, they demonstrated that under certain conditions, codes with all-symbol locality $r$
attaining the upper bound \eqref{eq:singleton} were not possible.
Song \etal{} then further extended this result \cite{Song.etal:2014}.

\begin{thm}[{\cite[Cor. 10]{Gopalan.etal:2012},\cite[Th. 10]{Song.etal:2014}}]\label{thm:structure}
Suppose $0< r<k\le n$, $r|k$ and $r+1\nmid n$. Then there exists no $[n,k,d]$ linear code with all-symbol locality $r$
with $d=n-k-k/r+2$.
\end{thm}


We present our key results in Theorem \ref{thm:main1} and \ref{thm:main2}, 
and defer their proofs to Section \ref{sec:proof}.
Our first construction provides a family of linear codes with addition based repair, 
information locality $r$ and attains distance one less than the value in \eqref{eq:singleton}.

\begin{thm}[Construction I]\label{thm:main1}
Fix $n<q$, $0< r<k\le n$ and $r|k$. Let $t=n-k-k/r$.
Then there exists an $[n,k,t+1]$ linear code over $\FF_q$ with addition based repair
and information locality $r$,
where $k/r$ parity symbols have locality $r$, and 
$t$ remaining parity symbols have locality $\min\{t-1,k\}$.
\end{thm}

Observe that in the case where $t\le r+1$, the code given by Theorem \ref{thm:main1} has all-symbol locality $r$.
Furthermore, when $t<r+1$, we have that $r+1\nmid n$ and so, Theorem \ref{thm:structure} applies and yields the following corollary.

\begin{cor}\label{cor:main}
Fix $n<q$, $0< r<k\le n$ and $r|k$. Suppose that $t=n-k-k/r\le r+1$.
There exists an $[n,k,t+1]$ linear code over $\FF_q$ with addition based repair
and all-symbol locality $r$.
Furthermore, if $t<r+1$, then the code is optimal in terms of distance.
\end{cor}

Our second construction provides a family of linear codes with addition based repair, 
all-symbol locality $r$ and attains the distance in \eqref{eq:singleton}.
\begin{thm}[Construction II]\label{thm:main2}
Fix $n<q$, $0< r<k\le n$ and $r|k$. 
Suppose that $n|r+1$, $r+1|q-1$, and $t=n-k-k/r$.
Then there exists an $[n,k,t+2]$ linear code over $\FF_q$ with addition based repair
and all-symbol locality $r$.
By \eqref{eq:singleton}, the code is optimal in terms of distance.
\end{thm}

\section{Code Constructions}
\label{sec:proof}

For the rest of the paper, we fix $n<q$, $0< r<k\le n$, $r|k$ and let $t=n-k-k/r$ and $m=k/r$.
In this section, we provide two constructions of codes with addition based repair.
\begin{enumerate}[(I)]
\item The first code construction simply requires $n<q$. Codes from this family have information locality $r$ and has distance $t+1$.
\item Codes in the second family have $r+1|n$ and we require $r+1|q-1$.
These codes have all-symbol locality $r$ and distance $t+2$.
\end{enumerate}

\subsection{Construction I}

To construct codes for Theorem \ref{thm:main1},
 we produce a $k\times n$ generator matrix $\vG$ of the form

{
\[
\left(\begin{array}{C{5mm}|C{5mm}|C{3mm}|C{5mm}|cccc}
&&&&																g_0^{(1)} & g_1^{(1)} & \cdots & g_{t-1}^{(1)}\\
\Big{\vM} & \Big{\vzero} & \makebox(0,-10){\ldots} & \Big{\vzero} &				g_0^{(2)} & g_1^{(2)} & \cdots & g_{t-1}^{(2)}\\
&& &&																\vdots & \vdots & \ddots & \vdots \\
&&&&									 							g_0^{(r)} & g_1^{(r)} & \cdots & g_{t-1}^{(r)}\\ \hline
&&&&																g_0^{(r+1)} & g_1^{(r+1)} & \cdots & g_{t-1}^{(r+1)}\\
\Big{\vzero} & \Big{\vM} & \makebox(0,-10){\ldots} & \Big{\vzero} &				g_0^{(r+2)} & g_1^{(r+2)} & \cdots & g_{t-1}^{(r+2)}\\
&& &&																\vdots & \vdots & \ddots & \vdots \\
&&&&									 							g_0^{(2r)} & g_1^{(2r)} & \cdots & g_{t-1}^{(2r)}\\ \hline
$\vdots$&$\vdots$&$\vdots$ &$\vdots$ &	\vdots & \vdots & \vdots & \vdots \\ \hline
&&&&																g_0^{(k-r+1)} & g_1^{(k-r+1)} & \cdots & g_{t-1}^{(k-r+1)}\\
\Big{\vzero} & \Big{\vzero} & \makebox(0,-10){\ldots} & \Big{\vM} &				g_0^{(k-r+2)} & g_1^{(k-r+2)} & \cdots & g_{t-1}^{(k-r+2)}\\
&& &&																\vdots & \vdots & \ddots & \vdots \\
&&&&									 							g_0^{(k)} & g_1^{(k)} & \cdots & g_{t-1}^{(k)}\\
\end{array}\right),
\]
}

\noindent where $\vM$ is the $r\times (r+1)$ matrix $(\vI_r|-\vj_r)$.
Here, $\vI_\ell$ and $\vj_\ell$ denote the identity matrix of dimension $\ell$
and the all-ones vector of length $\ell$, respectively.

Clearly, $\vG$ has full rank and hence, it remains to find the elements $g^{(i)}_j$ for $i\in[k]$ and $0\le j\le t-1$
such that $\vG$ generates a linear code $\C$ of distance $t+1$ with addition based repair.
To this end, we demonstrate that $\C$ is a punctured subcode of a cyclic code of length $q-1$ 
with the appropriate distance. 

Specifically, we identify a vector $(c_0,c_1,\ldots, c_{q-2})$ of length $q-1$
with the polynomial $c(x)=\sum_{j=0}^{q-2} c_jx^j$ in the ring $\FF_q^{q-1}/\langle x^{q-1}-1\rangle$.
Under this mapping of vectors to polynomials, 
let the rows of $\vG$ be mapped to polynomials $c^{(1)}(x)$, $c^{(2)}(x)$,\ldots, $c^{(k)}(x)$,
and so, the code $\C$ is identified with the vector space generated by $\{c^{i}(x): i\in[k]\}$.
Note that each codeword of $\C$ has length $n\le q-1$. 
However, we can regard them as vectors of length $q - 1$ by appending $q - 1 - n$ zeros to the right of each codeword
and this operation does not affect the polynomial representation of the codewords.
The next lemma provides sufficient conditions for $c^{(i)}(x)$ to generate 
a code with addition based repair and distance $t+1$.

\begin{lem}\label{lem:main}
Let $\omega$ be the primitive element in $\FF_q$.
Suppose that $\C$ is generated by $\{c^{(i)}(x): i\in[k]\}$ as defined above.
If $c^{(i)}(x)$ has roots at $1,\omega,\ldots,\omega^{t-1}$ for all $i\in[k]$,
then $\C$ has addition based repair and distance $t+1$.
\end{lem} 

To prove this lemma, we employ the  BCH bound on minimum distance of a cyclic code. 
We first recall some notation from ~\cite[Ch. 7]{MW_S}.
Let $\C$ be a cyclic code of length $q-1$ over $\FF_q$. 
An element $\alpha \in \FF_q$ is called a \emph{zero} of $\C$
if $c(\alpha) = 0$ for every codeword $c(x) \in \C$. 
Let $Z$ be the set of all zeros of $\C$. 
The polynomial $g(x) \triangleq \prod_{\alpha \in Z} (x-\alpha)$ is called the \emph{generator polynomial} of $\C$. 
Then $c(x) \in \C$ if and only if $g(x) | c(x)$. 

\begin{thm}[BCH bound] Let $\omega$ be the primitive element of $\FF_q$ and $t$ be an integer.
The cyclic code with the generator polynomial $(x-1)(x-\omega)\cdots(x-\omega^{t-1})$
has minimum distance at least $t+1$.
\end{thm}



\begin{proof}[Proof of Lemma \ref{lem:main}]
For the distance properties, we show that $\C$ is a subcode
of a cyclic code with generator polynomial $g(x)=(x-1)(x-\omega)\cdots(x-\omega^{t-1})$.
Then the distance of $\C$ is immediate from the BCH bound.

Indeed, for $i\in[k]$, since $c^{(i)}(x)$ has roots at $1,\omega,\ldots,\omega^{t-1}$,
$g(x)$ divides $c^{(i)}(x)$. Hence, $g(x)$ divides all polynomials in the span 
of $\{c^{(i)}(x): i\in[k]\}$, or, $\C$, and, the claim follows.

To demonstrate that $\C$ has addition based repair, 
we partition the first $n-t$ coordinates into $m$ groups of $r+1$ nodes
and place the last $t$ coordinates into one group.
Therefore, it suffices to show that the sum of the symbols in each group is zero
for a typical codeword in $\C$.
Equivalently, we verify this fact for all rows in the generator matrix $\vG$.

Now, the claim is clear for the first $m$ groups from the description of $\vG$.
For the last group of $t$ symbols, consider $i\in[k]$. 
Then $c^{(i)}(x)=x^{t_1}-x^{t_2}+x^{n-t}\sum_{j=0}^{t-1}g_j^{(i)}x^j$
for some $0\le t_1<t_2\le n-t-1$. 
Since $c^{(i)}(1)=0$, we have $\sum_{j=0}^{t-1}g_j^{(i)}=0$, as desired. 
\end{proof}

\vspace{1mm}

Therefore, to complete the proof of Theorem \ref{thm:main1},
we show that $g^{(i)}_j$ exist for $i\in[k]$ and $0\le j\le t-1$
and they can be found by solving $k$ linear systems,
each involving $t$ equations in $t$ variables.
Specifically, fix $i\in[k]$ and let $c^{(i)}(x)=x^{t_1}-x^{t_2}+x^{n-t}\sum_{j=0}^{t-1}g_j^{(i)}x^j$ for some $0\le t_1<t_2\le n-t-1$. 
Then $c^{(i)}(x)$ having roots at $1,\omega,\ldots,\omega^{t-1}$ is equivalent to the following linear system:
{
\begin{equation*}
\left(\begin{array}{cccc}
1 & 1 & \cdots & 1\\
\omega^{n-t} & \omega^{n-t+1} & \cdots & \omega^{n-1}\\
\vdots & \vdots & \ddots & \vdots \\
\omega^{(t-1)(n-t)} & \omega^{(t-1)(n-t+1)} & \cdots & \omega^{(t-1)(n-1)}\\
\end{array}\right) \vg
=\vw,
\end{equation*}
}
where $\vg=(g^{(i)})_{j=0}^{t-1}$ and $\vw=(\omega^{t_1j}-\omega^{t_2j})_{j=0}^{t-1}$.

Observe that the matrix representing the system of equations is Vandermonde.
Hence, we have a unique solution for $g_j^{(i)}$ and 
this completes the proof of Theorem \ref{thm:main1}.

\begin{exa}
Let  $n=11$, $q=13$, $r=3$ and $k=6$.
To compute the matrix $\vG$, we solve
six linear systems,
each involving three equations in three unknowns. Hence,
\[\vG=
\left(\begin{array}{cccc|cccc|ccc}
1 & 0 & 0 & 12 & 0 & 0 & 0 & 0 & 9 & 10 & 7 \\
0 & 1 & 0 & 12 & 0 & 0 & 0 & 0 & 12 & 4 & 10 \\
0 & 0 & 1 & 12 & 0 & 0 & 0 & 0 & 10 & 4 & 12 \\
0 & 0 & 0 & 0 & 1 & 0 & 0 & 12 & 11 & 2 & 0 \\
0 & 0 & 0 & 0 & 0 & 1 & 0 & 12 & 9 & 7 & 10 \\
0 & 0 & 0 & 0 & 0 & 0 & 1 & 12 & 11 & 8 & 7
\end{array}\right),
\]  
\noindent and $\vG$ generates an  $[11,6,4]$ linear code $\C$ 
with addition based repair. 
Nodes 1 to 8 have locality three, while 
nodes 9, 10, 11 have locality two.
Hence, $\C$ has all-symbol locality three. 
Since no $[11,6,5]$ linear code exists with all-symbol locality three by Theorem \ref{thm:structure},
we conclude that $\C$ is optimal. \hfill\qed
\end{exa}

\begin{rem} \label{rem:tamo}
We consider alternative approaches to the problem and point out certain pitfalls.
Suppose there exists an $[k+t-1,k,t]$ MDS code.
One common approach in the literature is to partition the $k+t-1$ nodes into
$m+1$ groups ($m$ groups of size $r$ and one group of size $t-1$),
and re-encode each group with either a $[r+1,r,2]$ or $[t,t-1,2]$ MDS code.
Equivalently, we add another $m+1$ nodes to obtain a code of length $n$ 
with the desired locality properties. 
The question is then whether the resulting code has addition based repair
and the desired distance.

\begin{enumerate}[(i)]
\item One obvious way to guarantee addition based repair is to
re-encode each group with MDS codes that have 
$(\vI_r|- \vj_r)$ and $(\vI_{t-1}|- \vj_{t-1})$ as their generator matrices.

Unfortunately, this method does not ensure that the resulting code is of distance $t+1$.
Consider the following generator matrix of a $[7,4,4]$ MDS code over $\FF_7$,
\[
\left(\begin{array}{cc|cc|ccc}
1 & 0 & 0 & 0 & 1 & 1 & 4 \\
0 & 1 & 0 & 0 & 1 & 2 & 3 \\
0 & 0 & 1 & 0 & 2 & 1 & 3 \\
0 & 0 & 0 & 1 & 2 & 6 & 5
\end{array}\right).\]
Adding three nodes as above yields the generator matrix,
\[
\left(\begin{array}{ccc|ccc|cccc}
1 & 0 & 6 & 0 & 0 & 0 &  1 & 1 & 4 & 1 \\
0 & 1 & 6 & 0 & 0 & 0 & 1 & 2 & 3 & 1\\
0 & 0 & 0 & 1 & 0 & 6 &  2 & 1 & 3 & 1\\
0 & 0 & 0 & 0 & 1 & 6 &  2 & 6 & 5 & 1
\end{array}\right).\]
However, the distance of the new code remains as four 
(distance between the first two rows is four).

\item When $t=r+1$, a more elaborate re-encoding was proposed by Tamo \etal{} \cite{Tamo.etal:2013}.
The initial MDS code in their method is chosen to be a Reed Solomon code defined over $\FF_q$,
while the groups are re-encoded using carefully chosen MDS codes defined over an extension field $\FF_{q^{k+1}}$.
Tamo \etal{} then demonstrated that the distance of the new code is increased to $t+2$, 
thus achieving the bound \eqref{eq:singleton}.

However, the repair algorithm for the code involves multiplications over a finite field.
This complexity is further increased by the fact that the underlying field is of the order $q^{k+1}$.
In general, multiplication over a binary extension field of order $2^s$ involves $O(s^{\log_2 3})$ bit operations
(see for example, \cite{Hankerson.etal:2004}). 
Thus, in the scheme of Tamo \etal{}, when the code is defined over $\FF_{q^{k+1}}$,
 the number of bit operations involved in the repair is increased by a factor of $(k+1)^{\log_2 3}$.
\end{enumerate}
%
\end{rem}

\begin{rem}
Recall below a code construction from Gopalan \etal{}~\cite{Gopalan.etal:2012}
with similar parameters as the codes in Corollary~\ref{cor:main}. 
\begin{quote}
\vspace{-3mm}
\begin{thm}[{\cite[Th. 15]{Gopalan.etal:2012}}]
\label{thm:Gopalan_15}
Let $0 < r < k \le n$ and $t = n - k - k/r < r + 1$. 
Suppose that $q > kn^k$. There exists a systematic $[n,k,t+2]_q$ code
$\C$ of information locality $r$, where $k/r$ parity symbols have locality $r$, and 
$t$ remaining parity symbols have locality $k-(k/r-1)(t-1)$. 
\end{thm}
\end{quote}

We compare our codes constructed in Corollary~\ref{cor:main} and 
the codes in Theorem~\ref{thm:Gopalan_15} in Table~\ref{tab:Gopalan_ours}.
Observe that our codes besides having addition based repair, also have smaller locality values and 
are defined over a significantly smaller field. However, 
the codes by Gopalan \etal{} have distance one more than our codes.

\begin{table}[H]%
\centering
\begin{tabular}{|l|c|c|}
\hline
&  Our code &  Gopalan \etal{} code\\
\hline
Addition based repair & Yes & No\\
\hline
Locality of $k + k/r$ nodes & $r$ & $r$\\
\hline
Locality of $t$ nodes & $t-1$ & $t-1 + (r+1-t)k/r$\\
\hline
Field size & $\mathcal{O}(n)$ & $\mathcal{O}(kn^k)$\\
\hline
Minimum distance & $t + 1$ & $t+2$\\
\hline
\end{tabular}
\vspace{2mm}

\caption{Comparison of Code from Corollary~\ref{cor:main} and Code from Theorem~\ref{thm:Gopalan_15} (by Gopalan \etal{}). 
}
\label{tab:Gopalan_ours}

\end{table}
\end{rem}

%

\subsection{Construction II}

We construct codes that satisfy the conditions in Theorem \ref{thm:main2}.
Here, we assume further $r+1|n$ and let $n=(k/r+\ell)(r+1)=(m+\ell)(r+1)$, in other words, $t=\ell r+\ell$.
Also suppose that $r+1|q-1$ and let $\omega$ be a primitive element of $\FF_q$.
Define $\alpha\triangleq \omega^{(q-1)/(r+1)}$,
and so, $\alpha$ has order $r+1$.

 In lieu of a generator matrix, we construct a {\em parity-check matrix}  for our code.
 Observe that $n-k=m+\ell+\ell r$.
 Hence, we consider the $(m+\ell+\ell r)\times n$ matrix $\vH$ of the form

{
\[
\vH\triangleq
\left(\begin{array}{cccc|cccc|c|cccc}
  1 & 1 & \cdots & 1 &
 0 & 0 & \cdots & 0 & \cdots &
 0 & 0 & \cdots & 0 \\
 0 & 0 & \cdots & 0 &
  1 & 1 & \cdots & 1 & \cdots &
 0 & 0 & \cdots & 0 \\
 \vdots & \vdots &\ddots & \vdots &
 \vdots & \vdots &\ddots & \vdots & \ddots &
 \vdots & \vdots &\ddots & \vdots \\
 0 & 0 & \cdots & 0 &
 0 & 0 & \cdots & 0 & \cdots &
   1 & 1 & \cdots & 1 \\ \hline
   
\multicolumn{4}{c|}{\vv_0^{(1)}}&
\multicolumn{4}{c|}{\vv_1^{(1)}}&
 \cdots &
\multicolumn{4}{c}{\vv_{m+\ell-1}^{(1)}}\\
\multicolumn{4}{c|}{\vv_0^{(2)}}&
\multicolumn{4}{c|}{\vv_1^{(2)}}&
 \cdots &
\multicolumn{4}{c}{\vv_{m+\ell-1}^{(2)}}\\
\multicolumn{4}{c|}{\vdots}&
\multicolumn{4}{c|}{\vdots}&
 \cdots &
\multicolumn{4}{c}{\vdots}\\
\multicolumn{4}{c|}{\vv_0^{(r)}}&
\multicolumn{4}{c|}{\vv_1^{(r)}}&
 \cdots &
\multicolumn{4}{c}{\vv_{m+\ell-1}^{(r)}}\\\hline

\multicolumn{4}{c|}{\vv_0^{(r+2)}}&
\multicolumn{4}{c|}{\vv_1^{(r+2)}}&
 \cdots &
\multicolumn{4}{c}{\vv_{m+\ell-1}^{(r+2)}}\\
\multicolumn{4}{c|}{\vv_0^{(r+3)}}&
\multicolumn{4}{c|}{\vv_1^{(r+3)}}&
 \cdots &
\multicolumn{4}{c}{\vv_{m+\ell-1}^{(r+3)}}\\
\multicolumn{4}{c|}{\vdots}&
\multicolumn{4}{c|}{\vdots}&
 \cdots &
\multicolumn{4}{c}{\vdots}\\
\multicolumn{4}{c|}{\vv_0^{(2r+1)}}&
\multicolumn{4}{c|}{\vv_1^{(2r+1)}}&
 \cdots &
\multicolumn{4}{c}{\vv_{m+\ell-1}^{(2r+1)}}\\\hline

\multicolumn{4}{c|}{\vdots}&
\multicolumn{4}{c|}{\vdots}&
 \ddots &
\multicolumn{4}{c}{\vdots}\\ \hline

\multicolumn{4}{c|}{\vv_0^{(\ell r+\ell-r)}}&
\multicolumn{4}{c|}{\vv_1^{(\ell r+\ell-r)}}&
 \cdots &
\multicolumn{4}{c}{\vv_{m+\ell-1}^{(\ell r+\ell-r)}}\\
\multicolumn{4}{c|}{\vv_0^{(\ell r+\ell-r+1)}}&
\multicolumn{4}{c|}{\vv_1^{(\ell r+\ell-r+1)}}&
 \cdots &
\multicolumn{4}{c}{\vv_{m+\ell-1}^{(\ell r+\ell-r+1)}}\\
\multicolumn{4}{c|}{\vdots}&
\multicolumn{4}{c|}{\vdots}&
 \cdots &
\multicolumn{4}{c}{\vdots}\\
\multicolumn{4}{c|}{\vv_0^{(\ell r+\ell-1)}}&
\multicolumn{4}{c|}{\vv_1^{(\ell r+\ell-1)}}&
 \cdots &
\multicolumn{4}{c}{\vv_{m+\ell-1}^{(\ell r+\ell -1)}}
  \end{array}\right),
\]
}
\noindent where \[\vv_i^{(j)}=((\omega^i)^j, (\omega^i\alpha)^j, \ldots, (\omega^i\alpha^r)^j) \mbox{ for }0\le i\le m+\ell-1, 0\le j\le n-1.\]

Observe that $0\le j\le n-1$, the vector $\left(\vv_i^{(j)}\right)_{i=0}^{m+\ell-1}$ is the evaluation of the polynomial $x^j$ 
with evaluation points $\{\omega^i\alpha^{i_1}: 0\le i\le m+\ell-1, 0\le i_1\le r\}$. Therefore, the set of $n$ vectors
$\vV\triangleq\left\{\left(\vv_i^{(j)}\right)_{i=0}^{m+\ell-1} : 0\le j \le n-1\right\}$ 
is linearly independent.

We first demonstrate the following lemma.

\begin{lem}\label{lem:samespace}
For $0\le i\le m+\ell-1, 0\le j\le n-1$, define $\vv_i^{(j)}$ as above. Then the span of $(m+\ell)$ vectors 
$\left\{\left(\vv_i^{(j(r+1))}\right)_{i=0}^{m+\ell-1} : 0\le j \le m+\ell-1\right\}$
is the same as the span of the $(m+\ell)$ vectors
\[\vJ\triangleq=\{(\vj_{r+1},\vzero,\ldots,\vzero),(\vzero,\vj_{r+1},\ldots,\vzero),\ldots,(\vzero,\vzero,\ldots,\vj_{r+1})\}.\]
\end{lem}

\begin{proof}
For $0\le j\le m+\ell-1$ and $0\le i\le m+\ell-1$, we observe that
\[\vv_i^{j(r+1)}=((\omega^i)^{j(r+1)}, (\omega^i\alpha)^{j(r+1)}, \ldots, (\omega^i\alpha^r)^{j(r+1)})=(\beta^{ij},\beta^{ij},\ldots,\beta^{ij})=\beta^{ij}\vj_{r+1},\]
where $\beta=\omega^{r+1}$.

 Therefore,
\begin{align*}
&\left(
\begin{array}{cccc}
\vv_0^{(0)} & \vv_1^{(0)} & \cdots &\vv_{m+\ell-1}^{(0)} \\
\vv_0^{(r+1)} & \vv_1^{(r+1)} & \cdots &\vv_{m+\ell-1}^{(r+1)} \\
\vdots & \vdots & \ddots & \vdots\\
\vv_0^{(m+\ell-1)(r+1)} & \vv_1^{(m+\ell-1)(r+1)} & \cdots &\vv_{m+\ell-1}^{(m+\ell-1)(r+1)} 
\end{array}
\right) \\
&= 
\left(
\begin{array}{cccc}
\vj_{r+1} &\vj_{r+1} & \cdots & \vj_{r+1} \\
\vj_{r+1} & \beta\vj_{r+1} & \cdots & \beta^{m+\ell-1}\vj_{r+1} \\
\vdots & \vdots & \ddots & \vdots\\
\vj_{r+1} &\beta^{(m-\ell-1)}\vj_{r+1} & \cdots & \beta^{(m-\ell-1)^2}\vj_{r+1} \\
\end{array}
\right)\\
&
= 
\left(
\begin{array}{cccc}
1&1&\cdots &1 \\
1 & \beta & \cdots & \beta^{m+\ell-1}\\
\vdots & \vdots & \ddots & \vdots\\
1 &\beta^{(m-\ell-1)} & \cdots & \beta^{(m-\ell-1)^2} \\
\end{array}
\right)
\left(
\begin{array}{cccc}
\vj_{r+1} &\vzero& \cdots & \vzero \\
\vzero & \vj_{r+1} & \cdots & \vzero \\
\vdots & \vdots & \ddots & \vdots\\
\vzero& \cdots & \vzero &\vj_{r+1}  \\
\end{array}
\right)
\end{align*}

Since 
$\left(\beta^{ij}\right)_{0\le i,j\le m-\ell-1}$
is invertible, we have the row spaces of the two matrices to be the same, as desired.
\end{proof}

Next, we demonstrate that the $\vH$ has full rank. In other words, the rank of $\vH$ is $m+\ell+\ell r$.
Notice that the row space of $\vH$ is given by the span of 
$\vJ\cup\left\{\left(\vv_i^{(j)}\right)_{i=0}^{m+\ell-1}:  0\le j\le \ell r+\ell-1, r+1\nmid j \right \}$.
From Lemma \ref{lem:samespace}, this space is given by $m+\ell+\ell r$ vectors from $\vV$.
Since these $n$ vectors are linearly independent, we have $\vH$ has full rank as desired.

Therefore, the code $\C$ with parity-check matrix $\vH$ has the desired dimension $k$.
Let us partition the $n$ coordinates into $m+\ell$ groups of size $r+1$.
Then from the first $m+\ell$ rows of $\vH$, we see any codeword in $\C$ has the property that  
the symbols in each group sum to zero. 
Hence, $\C$ has all-symbol locality $r$ with addition based repair.

It then remains to verify the distance of $\C$. 
Let $\vH'$ be the $(\ell + \ell r +1)\times n$ matrix that consists of the rows 
$\left\{\left(\vv_i^{(j)}\right)_{i=0}^{m+\ell-1}:  0\le j\le \ell r+\ell \right\}$.
It follows from Lemma \ref{lem:samespace} that the code generated by $\vH$ contains the code generated by $\vH'$.
So, $\C$ is contained in the code with parity-check matrix $\vH'$.

On the other hand, $\vH'$ is a generator matrix of a Reed-Solomon $[n,\ell +\ell r+1 ,n-\ell -\ell r]$ code.
Since the dual of an MDS code is also MDS,
the code with parity-check matrix $\vH'$ is a $[n,n-\ell -\ell r-1,\ell+\ell r+2]$ code.
Therefore, the code $\C$ has distance $\ell+\ell r+2=t+2$, as desired.

\begin{exa} Let $k=6$, $r=3$, and $n=12$.
Choose $q=13$ so that $r+1|q-1$ and $n|q-1$. Then $\vH$ is given by
\[       	
\left(\begin{array}{cccc|cccc|cccc}
1 & 1 & 1 & 1 & 0 & 0 & 0 & 0 & 0 & 0 & 0 & 0 \\
0 & 0 & 0 & 0 & 1 & 1 & 1 & 1 & 0 & 0 & 0 & 0 \\
0 & 0 & 0 & 0 & 0 & 0 & 0 & 0 & 1 & 1 & 1 & 1 \\
1 & 8 & 12 & 5 & 2 & 3 & 11 & 10 & 4 & 6 & 9 & 7 \\
1 & 12 & 1 & 12 & 4 & 9 & 4 & 9 & 3 & 10 & 3 & 10 \\
1 & 5 & 12 & 8 & 8 & 1 & 5 & 12 & 12 & 8 & 1 & 5
\end{array}\right),\]
\noindent and the linear code whose parity-check matrix is $\vH$
is given in Example \ref{exa:motivating}. 
\end{exa}

\begin{rem}\label{rem:tamo.barg} 
We observe certain similarities with the class of {\em Tamo-Barg} codes, 
proposed by Tamo and Barg \cite{Tamo.Barg:2014}. 
The Tamo-Barg codes may be viewed as a generalization of the Reed-Solomon codes.
As with Reed-Solomon code, each codeword in a Tamo-Barg code is 
the evaluation of a low degree polynomial over a subset of points in $\FF_q$.

Suppose\footnote{The condition of $r+1|n$ may be relaxed (see for example Constructions 5 and 6 in \cite{Tamo.Barg:2014}). 
However, we consider the case of $r+1|n$ for ease of exposition.} 
$r+1|n$. In general, a Tamo-Barg code requires a collection of $n/(r+1)$ disjoint groups $S_0,S_1,\ldots, S_{n/(r+1)-1}$
of evaluation points in $\FF_q$ and a polynomial $g(x)$ of degreee $r+1$ 
such that $g$ is constant on the points in the same group.
Using $g$, Tamo and Barg then choose the space of polynomials and 
evaluate them onto the points in $\bigcup_{i=0}^{n/(r+1)-1}S_i$ to obtain the codewords.
In some sense, Tamo and Barg constructs a {\em generator matrix} using $g$ and $S_s$.

In Construction II, we also have $n/(r+1)=m+\ell$ disjoint groups $S_i=\{\omega^i\alpha^j: 0\le j\le r\}$, $0\le i\le m+\ell-1$, and
we may choose $g(x)=x^{r+1}$. 
However, a fundamental difference in our approach is that 
we exploit $g$ and $S_i$ to construct a {\em parity-check matrix} for our codes.
Furthermore, if we apply Tamo and Barg's construction with the parameters $n=12$, $q=13$, $k=6$, $r=3$, 
we do not obtain a code with addition based repair.
\end{rem}

\section{Comparison with Previous Work}
\label{sec:comparison}

In this section, we compare the family of codes constructed in Theorem \ref{thm:main1} and \ref{thm:main2}
with certain families of optimal codes with information locality $r$.
Here, we focus on codes defined over fields of size $\mathcal{O}(n)$, 
where $n$ is the code length.
As pointed out in Remark \ref{rem:tamo}(ii), when the underlying field is large,
the complexity of finite field arithmetic is appreciably increased. 

On the other hand, when the underlying field is binary,
any linear code is naturally equipped with addition based repair.
Binary linear codes with good locality and distance properties were first studied
by Goparaju and Calderbank \cite{Goparaju.Calderbank:2014} and
their work was later extended by Tamo \etal{} \cite{Tamo.etal:2015} and Zeh and Yaakobi \cite{Zeh.Yaakobi:2015}.
However, optimal codes with all-symbol locality $r$ and distance $d$
are only known for a limited set of parameters 
and other exceptional cases \cite{Goparaju.Calderbank:2014,Zeh.Yaakobi:2015}.

Hence, we compare our construction with the following two classes of optimal codes with information locality $r$:
{\em Pyramid codes}, proposed by Huang \etal{} \cite{Huang.etal:2013}
and {\em Tamo-Barg codes}. 
Table \ref{tab:comparison} summarizes the differences in repair complexities, locality and distance properties,
and the analysis is detailed below.

We observe that the codes from Theorem \ref{thm:main1} and \ref{thm:main2} are most efficient in terms of repair complexities.
However, the codes from Theorem \ref{thm:main1} are slightly weaker in terms of distance.
\vspace{1mm}

\noindent{\bf Pyramid Codes}. Pyramid codes are a family of erasure codes with information locality $r$ that attain the bound \eqref{eq:singleton}
and one variant is currently deployed in 
Windows Azure System \cite{Huang.etal:2012}.
%
%
In general, a Pyramid code is a linear $[n,k,t+2]$ code whose $k$ information nodes
are divided into $m$ groups $S_1, S_2,\ldots, S_m$ of size $r$. 
For each group $S_i$ of $r$ information nodes $x_1^{(i)},x_2^{(i)},\ldots,x_r^{(i)}$, 
we associate a parity node $y$ whose value is given by 
\[y^{(i)}=\sum_{j=1}^r \alpha^{(i)}_j x_j^{(i)} \mbox{ for some nonzero values }\alpha_1,\alpha_2,\ldots,\alpha_r. \] 
 Observe that to recover $x_\ell^{(i)}$, $\ell\in[r]$, we compute
 \[x^{(i)}_\ell=\left(\alpha^{(i)}_\ell\right)^{-1}\left(y-\sum_{j\ne \ell}\alpha^{(i)}_j x^{(i)}_j\right).\]
 In either case, to repair any of these $m(r+1)$ nodes, we require $\mathcal{O}(r)$ finite field operations that include multiplications.
 
 For each of the remaining $t$ parity nodes, say $z_\ell$, $\ell\in[t]$, its value is given by
 \[z_\ell=\sum_{i=1}^{m}\sum_{j=1}^r \beta_j^{(i)} x_j^{(i)} \mbox{ for some $k$ nonzero values }\beta_j^{(i)}. \] 
Hence, the repair of each of these $t$ parity nodes require $\mathcal{O}(k)$ finite field operations that include multiplications.
\vspace{1mm}

\noindent{\bf Tamo-Barg Codes}.
%
 Suppose $r+1|n$. Recall that a Tamo-Barg code requires a collection of $n/(r+1)$ disjoint groups $S_0,S_1,\ldots, S_{n/(r+1)-1}$
of evaluation points in $\FF_q$. 
 Then the repair function for node $\alpha\in S_s$ is given by
 \[ c_\alpha=\sum_{\beta\in S_s\setminus\{\alpha\}} c_\beta \prod_{\beta'\in S_s\setminus\{\alpha,\beta\}} \frac{\alpha-\beta'}{\beta-\beta'}. \]
Therefore, the repair of node $\alpha$ requires $\mathcal{O}(r^2)$ finite field operations that include multiplications.

\begin{table*}
\caption{Comparisons of Codes with Good Locality and Distance Properties 
\vspace{-2mm}
}
\label{tab:comparison}
Codes of length $n$, dimension $k$ and information locality $0< r<k\le n$.
Assume also, $r|k$ and let $t=n-k-k/r$. 
\vspace{1mm}

\begin{tabular}{|l|p{6cm}|p{6.5cm}|p{1.3cm}|}
\hline
& Repair Complexity & Locality & Distance\\
\hline
\hline
Theorem \ref{thm:main1} &
$\mathcal{O}(r)$ additions only for $\frac{k}{r}(r+1)$ nodes &
$\frac{k}{r}(r+1)$ nodes with locality $r$  &
$t+1$\\
&$\mathcal{O}(t)$ additions only for $t$ nodes & $t$ parity nodes with locality $\min\{t-1,k\}$ &\\
\hline
Theorem \ref{thm:main2} &
$\mathcal{O}(r)$ additions only for all nodes &
All-symbols locality $r$ &
$t+2$\\
\hline
Pyramid Codes &
$\mathcal{O}(r)$ field operations for $\frac{k}{r}(r+1)$ nodes &
$\frac{k}{r}(r+1)$ nodes with locality $r$ &
$t+2$\\
&$\mathcal{O}(k)$ field operations for $t$ nodes & $t$ parity nodes with locality $k$ &\\

\hline
Tamo-Barg &
$O(r^2)$ field operations for all nodes &
All-symbols locality $r$ &
$t+2$\\
\hline
\end{tabular}
\vspace{-5mm}
\end{table*}

\section{Concluding Remarks}
\label{sec:conclusion}

We introduced the notion of addition based repair for locally repairable codes
in order to accelerate the repair process for single node failures.
Furthermore, since the repair involves only addition operations,
implementation of these codes at the hardware level is considerably easier.

The theoretical results are encouraging. 
We constructed two families of codes with information locality $r$ and addition based repair, 
whose distances come close to the Singleton-like bound \eqref{eq:singleton}.
In addition, the codes in Construction II bear a resemblance to the Tamo-Barg codes (see Remark \ref{rem:tamo.barg})
and it will be interesting to derive certain connections between the two code families.

Another direction of study is to allow other efficient field operations, 
such as the bit-wise shift operation for binary extension fields, in the repair process.
As in the study of array codes \cite{Blaum.Roth:1993} and regenerating codes \cite{Hou.etal:2013,Shum.etal:2014}, 
the introduction of these operators may be employed to improve the encoding and decoding complexities for codes with local repair.

\bibliographystyle{IEEEtran}
\bibliography{LRC}

\end{document}